\documentclass[aps.pra,onecolumn,showpacs,preprintnumbers,amsmath,amssymb]{revtex4-2}
    
    \usepackage[utf8]{inputenc}
    \usepackage[english]{babel}
    \newtheorem{theorem}{Theorem}
    
    \newtheorem{lemma}{Lemma}
    \usepackage{float}
    \newtheorem{definition}{Definition}
    \newtheorem{remark}{Remark}
    \usepackage{setspace}
    \usepackage{booktabs}
    \usepackage{multirow}
    \usepackage{amsmath,amssymb}
    \usepackage{qcircuit}
    \usepackage{caption}
    \usepackage{lipsum}
    \usepackage{tabularx}
    \usepackage{natbib}
    \usepackage{stfloats}
    \usepackage{graphicx}
    \usepackage{rotating}
    \usepackage{makecell}
    \setcellgapes{2.3pt}
    
    \usepackage{subcaption} 
    \newenvironment{proof}{\textit{Proof.}}{\hfill$\square$}
    \usepackage{physics}
    \usepackage[figurename=Fig., justification=RaggedRight]{caption}
    \usepackage{amsmath}
    \usepackage{xcolor}
    \usepackage{hyperref}

    \usepackage{adjustbox,expl3,etoolbox}
    \letcs\replicate{prg_replicate:nn}
    
\begin{document} 
\title{Efficient explicit circuit for quantum state preparation of piecewise continuous functions}

\begin{abstract}
Efficiently uploading data into quantum states is essential for many quantum algorithms to achieve advantage across various applications. In this paper, we address this challenge by developing a method to upload a polynomial function $f(x)$ on the interval $x \in [-1,1]$ into a pure quantum state consisting of qubits, where a discretized $f(x)$ is the amplitude of this state. The preparation cost has $\mathcal{O}(n\log n)$ scaling in the number of qubits $n$ and linear scaling with the degree of the polynomial $Q$. This efficiency allows the preparation of states whose amplitudes correspond to high-degree polynomials (up to $10^4$), enabling accurate approximation of functions that admit efficient polynomial series representations and whose amplitude profiles are not extremely localized. We provide a fully explicit circuit realization, based on four real polynomials that meet specific parity and boundedness conditions. We extend this construction to cover piece-wise polynomial functions, a case not previously addressed explicitly in the literature, the algorithm scaling linearly with the number of piecewise parts. Our method achieves efficient quantum circuit implementation and we present detailed gate counting and resource analysis.
\end{abstract}

\author{Nikita Guseynov}
\email{guseynov.nm@gmail.com}
\affiliation{Global College, Shanghai Jiao Tong University, Shanghai 200240, China.}

\author{Nana Liu}
\email{nana.liu@quantumlah.org}
\affiliation{Institute of Natural Sciences, School of Mathematical Sciences, Shanghai Jiao Tong University, Shanghai 200240, China}
\affiliation{Ministry of Education Key Laboratory in Scientific and Engineering Computing, Shanghai Jiao Tong University, Shanghai 200240, China}
\affiliation{Shanghai Artificial Intelligence Laboratory, Shanghai, China}
\affiliation{Global College, Shanghai Jiao Tong University, Shanghai 200240, China}

 \maketitle
 
\section{Introduction}

Quantum computing \cite{FEYNMAN} holds significant promise for revolutionizing various fields such as
material science \cite{aspuru2005simulated}, computer science \cite{costa2019quantum,gaitan2020finding,jin2023time,linden2022quantum,jin2022quantum}, finance \cite{stamatopoulos2020option,gonzalez2023efficient} due to its potential to solve some
problems more efficiently than canonical methods on classical computers. The development of
efficient quantum algorithms is crucial to harnessing this power, enabling quantum
computers to outperform their classical counterparts \cite{zhang2022brief,nielsen2002quantum}. However, the effectiveness of
these algorithms can be severely undermined by inefficient data loading processes \cite{tang2021quantum}.

In this paper, we present an explicit and resource-transparent method to upload a continuous  
or piece-wise continuous function $f(x)$ defined on the interval  
$x \in [-1,1]$ (where $x$ is discretized along this interval  
with $\Delta x \sim 1/N$) into a digital qubit-based quantum  
state as the amplitudes of a normalized vector $\ket{\psi}^n$  
\begin{eqnarray}
\left| \psi \right\rangle^n := \frac{1}{N_f} \sum_{x=0}^{2^n-1} f(x_i) |x\rangle^n,
\label{eq:eq_1}
\end{eqnarray}
with a size $N=2^n$, where $n$ is the number of qubits.  
The domain $[-1,1]$ can always be generalized to any interval  
$[a,b]$ by a simple rescaling of $x$.

\begin{table}[H]
\centering
\begin{small}
\begin{tabular}{c| c| c |c |c}
\textbf{Name} & \textbf{Notation} & \textbf{Polynomial Degree} & \textbf{Gate Counts} & \textbf{Domain} \\
\hline
Exponent & $e^{\alpha x}$ & $\mathcal{O}\left(\frac{|\alpha| + \log(1/\epsilon)}{\log\left( 1+\frac{\log(1/\epsilon)}{|\alpha|} \right)}\right)$ & $\mathcal{O}\left(n\log n \cdot \frac{|\alpha| + \log(1/\epsilon)}{\log\left( 1+\frac{\log(1/\epsilon)}{|\alpha|} \right)}\right)$ & $x \in [-1, 1];\alpha\in \mathbb{C}$  \\
Cosine & $\cos(tx)$ & $\mathcal{O}\left( t + \frac{\log(1/\epsilon)}{\log\left(e + \frac{\log(1/\epsilon)}{t}\right)} \right)$ & $\mathcal{O}\left(n\log n \left( t + \frac{\log(1/\epsilon)}{\log\left(e + \frac{\log(1/\epsilon)}{t}\right)} \right)\right)$ & $x \in [-1, 1];t\in \mathbb{R}$ \\
Sine & $\sin(tx)$ & $\mathcal{O}\left( t + \frac{\log(1/\epsilon)}{\log\left(e + \frac{\log(1/\epsilon)}{t}\right)} \right)$ & $\mathcal{O}\left(n\log n \left( t + \frac{\log(1/\epsilon)}{\log\left(e + \frac{\log(1/\epsilon)}{t}\right)} \right)\right)$ & $x \in [-1, 1];t\in \mathbb{R}$ \\
Sigmoid & $\frac{1}{1 + e^{-2x}}$ & $\mathcal{O}\left(\log \left(1/\epsilon\right)\right)$ & $\mathcal{O}\left(n\log n \cdot\log \left(1/\epsilon\right)\right)$ & $x\in[-\pi/2,\pi/2]$ \\
Gaussian & $\frac{1}{\sigma \sqrt{2\pi}} e^{-\frac{1}{2}(x/\sigma)^2}$ & $\mathcal{O}\left(\frac{ \log\left(1/\epsilon\right) + \log \left(1/\sigma\right) }{ \log\left(1+2\sigma\log\left(1/\epsilon\right)\right)}\right)$ & $\mathcal{O}\left(n\log n\cdot\frac{ \log\left(1/\epsilon\right) + \log \left(1/\sigma\right) }{ \log\left(1+2\sigma\log\left(1/\epsilon\right)\right)}\right)$ & $x \in [-1, 1];\sigma>0$ \\
Bessel & $J_n(\alpha x)$ & $\mathcal{O}\left(\frac{ \log\left( 1/\epsilon \right) }{ \log\left( 1+\frac{2}{|\alpha|}\log\left( 1/\epsilon\right)  \right) }
\right)$ & $\mathcal{O}\left(n\log n \cdot \frac{ \log\left( 1/\epsilon \right) }{ \log\left( 1+\frac{2}{|\alpha|}\log\left( 1/\epsilon\right)  \right)}\right)$ & $x \in [-1, 1];\alpha\in \mathbb{R}$ \\
Reciprocal & $\frac{1}{x}$ & $\mathcal{O}\left(\delta \sqrt{\log^2\left(1/\epsilon\right) + \log\left(1/\epsilon\right) \log(\delta)}\right)$ & $\mathcal{O}\left(\delta n\log n \cdot  \sqrt{\log^2\left(1/\epsilon\right) + \log\left(1/\epsilon\right) \log(\delta)}\right)$ & $x \in [-1, -\frac{1}{\delta}] \cup [\frac{1}{\delta}, 1]$ \\
ReLU & $\text{ReLU}(x)$ & $\mathcal{O}(1)$ & $\mathcal{O}(n\log n)$ & $x \in \mathbb{R}$ \\
Leaky ReLU & $\text{Leaky ReLU}(x)$ & $\mathcal{O}(1)$ & $\mathcal{O}(n\log n)$ & $x \in \mathbb{R}$ \\
\end{tabular}
\end{small}
\caption{
  Polynomial Degree and Gate Counts for Various Functions \cite{childs2017quantum,abramowitz1968handbook}. The polynomial degree means that we retain only $k$ terms from the expansion for the function; the gate counts directly follow from the series expansions from Table~\ref{table:1} and Theorem~\ref{theorem: Linear combination of polynomials}. The main discretization parameter here is $n=\log_2N$ which is the number of qubits. The error term $\epsilon$ represents the maximum absolute value of the residual beyond $k$ terms in the series expansion across the given domain. The $\delta$ value for the Reciprocal functions is defined by its domain. For ReLU and Leaky ReLU we suggest using the piece-wise approach from the Appendix~\ref{appendix:piece-wise}.
}
\label{table:poly_degree_gate_counts}
\end{table}

Quantum states of this kind play a key role in numerous  
quantum algorithms. They are used in discrete-time quantum walk approaches~\cite{childs2010relationship}, and modeling electromagnetic scattering cross  
sections~\cite{scherer2017concrete}. They also appear in quantum algorithms for  
solving linear and nonlinear partial differential equations~\cite{cao2013quantum,leyton2008quantum},  
and in specifying priors for quantum phase estimation~\cite{berry2024analyzing}.  

While the digital quantum circuit model offers a universal platform,  
the fundamental operations available on current quantum hardware are  
typically restricted to one- and two-qubit gates, such as C-NOTs and single-qubit
rotations, which play a role analogous to the NOT and OR gates in  
classical computing \cite{nielsen2002quantum}. A nontrivial challenge in this framework is to  
determine how, starting from the initial state $\ket{0}^n$, one can  
efficiently combine only these fundamental gates to prepare a state  
whose amplitudes encode a prescribed set of function values, as in  
Eq.~(\ref{eq:eq_1}). This "uploading" task is far from automatic:  
it requires decomposing the functional map into a concrete gate sequence, while optimizing for circuit depth and hardware constraints. The difficulty is particularly pronounced for large $n$, where the direct uploading procedure would require storing $2^n$-array of data—an exponential classical cost. The subsequent quantum brute-force routine \cite{araujo2023configurable,biamonte2017quantum} to upload that array has gate cost
$\mathcal{O}(\sqrt{2^n})$ which is inefficient.

Since polynomials are efficiently computable classically, one might expect  
quantum arithmetic circuits to be the natural choice for state preparation.  
However, in practice, such circuits are significantly more expensive in  
terms of gate complexity and ancilla requirements. In addition, quantum  
arithmetic algorithms often require the implementation of functions such  
as the arccosine, which is technically challenging and resource-intensive. To avoid this costly overhead, we present a method that directly constructs  
state preparation unitaries using the polynomial structure itself, without  
relying on arithmetic subroutines~\cite{munoz2018t,haner2018optimizing}, offering practical circuit  
synthesis for encoding piecewise continuous functions into quantum states.

The resource cost of our method scales as $\mathcal{O}(Q n \log n)$,  
where $n$ is the number of qubits and $Q$ is the degree  
of the polynomial. This linear dependence on \(Q\) enables the efficient construction of
high-degree polynomials, allowing us to upload polynomial approximations of
functions with known series expansions (such as Taylor or Jacobi–Anger
expansions \cite{abramowitz1968handbook}) in regimes where these series are efficient and the filling
ratio \(\mathcal{F}\) remains appreciable. Furthermore, we extend the construction to accommodate piecewise polynomial functions while preserving  
the same scaling behavior. For a function divided into $G$ segments,  
the total resource scaling becomes $\mathcal{O}(Q n \log n+Gn+QG)$.

To achieve this, we approximate the target complex function as 
a superposition of four real polynomials with different parities. For each 
of those polynomial we build an independent state preparation unitary using 
the Alternating Phase Transition Sequence (APTS) \cite{gilyen2019quantum}, 
also known as qubitization \cite{dalzell2023quantum}. This procedure applies 
the target polynomial function on the eigenvalues $\lambda_i$ of a block-encoded 
simple matrix. In particular, we use the matrix with linear eigenvalues 
$\lambda_i\sim Ai+B$ as an input for the APTS. Finally, we use the four prepared block-encoded polynomials as an input for a linear combination of unitaries (LCU) \cite{dalzell2023quantum} to embed those polynomials into the desired function.

In addition, we suggest a method to construct a piecewise polynomial function, defined as an independent polynomial of different degrees in each sector. To realize this in practice, we developed an efficient procedure that uses an ancilla qubit as an indicator for each segment $-1 \leq a_1 \leq b_1 \leq 1$ of the domain $[-1, 1]$. For a given computational basis state $\ket{i}^n$, the unitary gate sets this ancilla to $|1\rangle$ if and only if the state belongs to the corresponding segment $x_i \in [a_1, b_1]$, effectively marking the target region. We then apply the controlled version of the state preparation algorithm described above, ensuring that the appropriate polynomial is uploaded only for the predetermined segment. After completing the segment-specific operation, we uncompute the indicator qubit, returning it to $|0\rangle$ so it can be reused later. This method enables independent polynomials of arbitrary degree to be assigned to each segment and preserves the linear scaling in terms of the polynomial degree $Q$.

This document is organized as follows: Section~\ref{section:preliminaries} provides the theoretical foundations, including definitions and essential mathematical tools. In Section~\ref{section: problem statement}, we present the main results of the paper. Section~\ref{section: related work} reviews alternative methods from the literature and compares them with our approach. Appendices~\ref{appendix: cooridnate-polynomial-oracle}, and \ref{appendix:piece-wise} detail our explicit algorithm for function uploading, outlining each step of the process.

\section{Preliminaries}\label{section:preliminaries}

Now we introduce some preliminary mathematical and quantum computing concepts that we will use, and the system of indexes we will employ. We start with defining the multi-control operator. 
    
    \begin{definition}[Multi-control operator]
    \label{def:multiconrol operator}
        Let $U$ be an $m$-qubit quantum unitary and $b$ a bit string of length $n$. We define $C_U^b$ as an $n+m$-qubit quantum unitary that applies $U$ to an $m$-qubit quantum register if and only the $n$-qubit quantum register is in state $\ket{b}$. In the Appendix~\ref{appendix: multi-control} we show how this gate can be implemented using $C^1_U$ operator, $16n-16$ single-qubit operations, and $12n-12$ C-NOTs.
        \[ C^b_U=\ket{b}^n\bra{b}^n\otimes U+\sum\limits_{\substack{i=0,\dots,2^n-1\\ i\neq b}}\ket{i}^n\bra{i}^n\otimes I^{\otimes m}  \]
    \end{definition}
   
    \begin{definition}[Pure ancilla]
    \label{def:pure ancilla}
    Let $U$ be an $n+m$-qubit unitary operator such that for any arbitrary $\ket{\psi}^n$
    \begin{equation}
        U\ket{0}^m\ket{\psi}^n = \ket{0}^m\ket{\phi}^n,
        \label{eq:aux_pure_ancilla}
    \end{equation}
    where $\ket{\phi}^n$ is some quantum state. Then we say that the $m$-qubit quantum register is pure ancilla for operator $U$, if before and after the action it is in the state $\ket{0}^m$.

    We note that from Eq.~(\ref{eq:aux_pure_ancilla}), there must exist an $n$-qubit unitary $V$ such that $\ket{\phi}^n = V\ket{\psi}^n$, i.e., on the subspace where the ancilla is in the $\ket{0}^m$ state, $U$ acts as $I^m \otimes V$.

    Throughout this paper, in all figures for quantum schemes, we depict wires corresponding to such pure ancillas as green dash-dot lines.
\end{definition}

    \begin{remark} 
    A pure ancilla is simply a qubit that is required to be present in the quantum computer for intermediate computation but does not need to be measured at the end of the algorithm. Since it is returned to the $\ket{0}$ state after use, it can be safely reused in subsequent computations or circuits without resetting. In Fig.~\ref{fig:Multi_control_qc} (c) we use $m-2$ auxiliary qubits for efficient C-NOT construction. Exploitation of this qubits is divided into two epochs: (i) entangling stage that allows us to achieve some sophisticated quantum states, (ii) purification stage that sets ancillas back to zero-state.
    \end{remark}

 \begin{definition}[Block-encoding(modified Definition 43 from \cite{gilyen2019quantum})]
        Suppose that $A$ is an $n$-qubit operator, $\alpha,\epsilon\in\mathbb{R}_+$; $a,s\in\mathbb{Z}_+$, and let $(a+s+n)$-qubit unitary $U$ be so that for any arbitrary quantum state $\ket{\psi}^n$
        \[ U\ket{0}^a\ket{0}^s\ket{\psi}^n=\ket{0}^a\ket{\phi}^{n+s};\qquad (\bra{0}^s\otimes I^{\otimes n})\ket{\phi}^{n+s}=\tilde{A}\ket{\psi}^n,\]
            then we say that $U$ is an $(\alpha,s,\epsilon)$-block-encoding of $A$, if
        \[ ||A-\alpha \tilde{A}||\leq\epsilon,\]
        where the $s$-qubit register, referred to as the workspace ancilla (Definition~\ref{definition:workspace_ancilla}), must be measured in the zero state to achieve the desired action on the system; if any other outcome is observed, the result is discarded.
        
        We underline that the unitary $U$ can exploit some auxiliary qubits setting them back to zero-state. Later we call such qubits  `pure ancillas' (Definition~\ref{def:pure ancilla}).  The general scheme of block-encoding is presented in Fig.~\ref{fig:block_encoding_general}.
    \end{definition}
    
\begin{figure}[H]
\centering
    \includegraphics[width=0.4\textwidth]{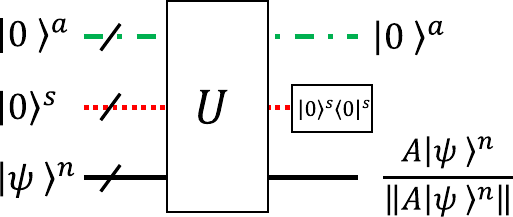}
    \caption{
        The general scheme of block-encoding for matrix $A$. The last operation on the second register indicates measurement of zero, projecting the main $n$-qubit register onto the desired state. 
        Wires for pure ancilla (Definition \ref{def:pure ancilla}) are depicted as green dash-dotted lines; workspace ancilla (Definition \ref{definition:workspace_ancilla}) are shown as red dotted wires; and the system register is represented by a solid black wire.
    }
    \label{fig:block_encoding_general}
\end{figure}

\begin{definition}[Workspace ancilla]\label{definition:workspace_ancilla}
Workspace ancilla are qubits initialized in $|0\rangle$ and  
added to the system register in block-encoding. They enable embedding  
a non-unitary operator $A$ as the top-left block of a unitary $U$:
\[
U = \begin{pmatrix}
A & * \\
* & *
\end{pmatrix}
\]
After $U$ is applied, workspace ancilla are not required to return to  
$|0\rangle$ and typically end in an unknown state (unlike pure ancilla).  
The action of $A$ is realized by projecting the workspace ancillas onto  
$|0\rangle^{\otimes a}$:
\[
\left( \langle 0|^{\otimes a} \otimes I \right) U \left( |0\rangle^{\otimes a} \otimes |\psi\rangle \right) = A|\psi\rangle.
\]
On quantum hardware, this corresponds to measuring all workspace ancillas  
and accepting only the all-zero result; otherwise, the system register  
state is discarded.

Throughout this paper, in all figures for quantum schemes, we depict wires corresponding to workspace ancillas as red dotted wires.
\end{definition}

\begin{definition}[Filling ratio \cite{rattew2022preparing}]\label{definition: filling ratio}
Given the usual definition for \( L^p \) function norms of a Riemann integrable, bounded function \( f \), \(\|f\|_p = \left[ \int_a^b |f(x)|^p dx \right]^{1/p}\), we define the \emph{filling ratio} of \( f \),

\[
\mathcal{F} := \frac{\|f\|^2_2}{ 2\|f\|_{\text{max}}^2}.
\]

Note that for continuous functions over a closed interval \(\|f\|_\infty =: \|f\|_{\text{max}}\) represents the absolute largest value of the function. The filling ratio has a clear geometrical interpretation shown in Fig.~\ref{fig:filling ratio}.

\begin{figure}[H]
\centering
    \includegraphics[width=0.5\textwidth]{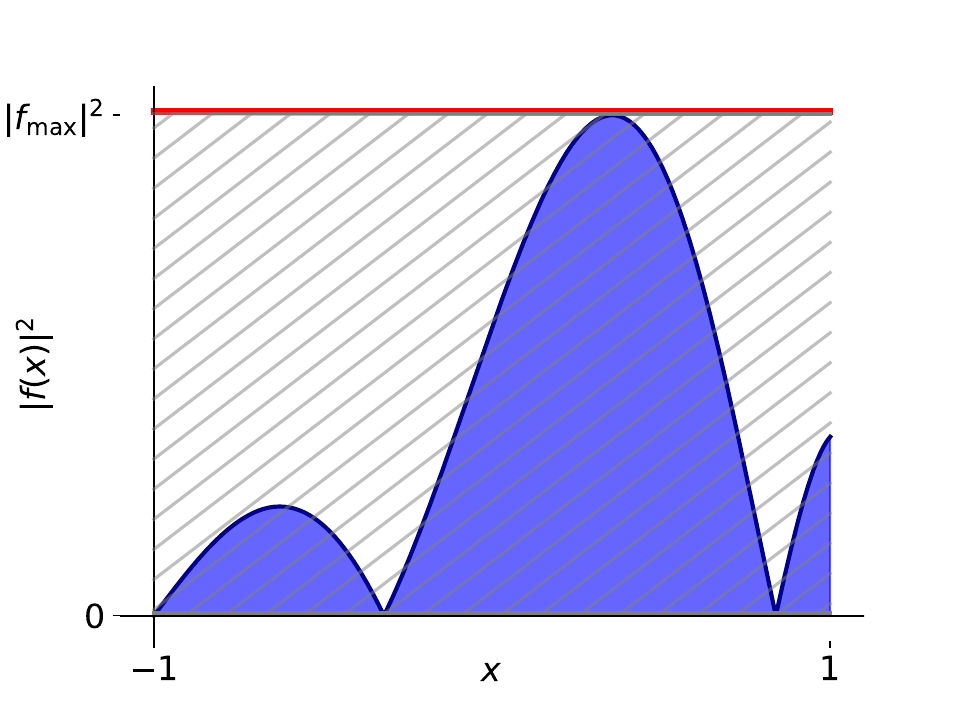}
    \caption{The geometrical interpretation of the filling ratio. The filling ratio is the area under the function $\abs{f(x)}^2$ (blue area) relative to its doubled absolute maximum value $2\abs{f_{\max}}^2$ (area in the dashed rectangle).}
    \label{fig:filling ratio}
\end{figure}
\end{definition}

\section{Results}\label{section: problem statement}

We begin with a continuous function $f(x)$ defined on the interval $[-1,1]$; $f: \mathbb{R} \rightarrow \mathbb{C}$. We assume that it has a  polynomial decomposition
\begin{equation}
    f(x)=\sum_{i=0}^Q\sigma_ix^i;\qquad \sigma_i\in\mathbb{C}.
    \label{eq:function_def}
\end{equation}
Polynomial decomposition is a crucial tool in both theoretical and applied mathematics due to its flexibility and effectiveness in approximating complex functions. As shown in the Table~\ref{table:poly_degree_gate_counts} and Table~\ref{table:1}, various functions can be effectively approximated using polynomial series expansions, each with specific error bounds and applicable domains.

Our main goal in this paper is to upload such a function into a digital qubit-based quantum computer as a normalized vector $\ket{\psi}^n\in\mathbb{C}^{2^n}$, defined as
\begin{equation}
    \ket{\psi}^n=\frac{\sum_{k=0}^{2^n-1} \sum_{i=0}^Q \sigma_ix^i_k\ket{k}^n}{\sqrt{\sum_{\kappa=0}^{2^n-1}\abs{\sum_{j=0}^Q\sigma_jx^j_\kappa}^2}},
    \label{eq:wave-function representation}
\end{equation}
which is a discrete version of Eq.~(\ref{eq:function_def}) with partitioning size $\Delta x=\frac{2}{2^n-1}$. Thus, given the discrete representation $x_k=-1+k\Delta x$, we can rewrite Eq.~(\ref{eq:wave-function representation}) as
\begin{equation}
    \ket{\psi}^n=\frac{\sum_{i=0}^Q\sum_{k=0}^{2^n-1}\sigma_i(a+k\Delta x)^i\ket{k}^n}{\sqrt{\sum_{\kappa=0}^{2^n-1}\abs{\sum_{j=0}^Q\sigma_j(a+\kappa\Delta x)^j}^2}}.
    \label{eq:wave-function representation_discrete from a to b}
\end{equation}

Our next step is to show that any polynomial of the form in Eq.~(\ref{eq:function_def}) can be written as a sum of four polynomials $P_s$ with the following properties:
\begin{enumerate}
    \item $P_{s}(x)$ has parity ($0$ or $1$ mod 2), which means that it has only odd or even degrees of $x$;\label{cond1}
    \item $P_{s}(x)$ is a real polynomial, which means that all the coefficients are real;\label{cond2}
    \item For all $x\in[-1,1]$: $\abs{P_{s}(x)}<1$.\label{cond3}
\end{enumerate}
Thus, we decompose Eq.~(\ref{eq:function_def}) into four real polynomials
\begin{equation}
    f(x)=\alpha_1P_1(x)+\alpha_2P_2(x)+i\alpha_3P_3(x)+i\alpha_4P_4(x);\qquad \alpha_s\in\mathbb{R},
    \label{eq:four_polyn}
\end{equation}
where $P_1(x)$ and $P_3(x)$ have parity 1, and $P_2(x)$ and $P_4(x)$ have parity 0. The third condition can be met by tuning $\alpha_i$. For such $f(x)$, we provide an explicit algorithm for preparing pure quantum states whose amplitudes are proportional to $f(x)$. We will also provide the exact gate count and show explicitly the required quantum circuits. The following theorems are the central theorems in this paper.

\begin{theorem}[Linear combination of polynomials]\label{theorem:intro 1}
Let $f(x)$ be a continuous function $f: \mathbb{R} \rightarrow \mathbb{C}$ with polynomial decomposition
\[f(x)=\sum_{i=0}^Q\sigma_ix^i=\alpha_1P_1(x)+\alpha_2P_2(x)+i\alpha_3P_3(x)+i\alpha_4P_4(x);\qquad \alpha_s\in\mathbb{R};\quad \sigma_i\in\mathbb{C},\]
where polynomials $P_s(x)$ satisfy the conditions \ref{cond1}, \ref{cond2}, \ref{cond3}. Then we can construct a unitary operation $U_f$ that prepares a $2^n$-discrete version of $f(x)$
\[\ket{\psi}^n=\frac{\sum_{i=0}^Q\sum_{k=0}^{2^n-1}\sigma_ix^i_k\ket{k}^n}{\sqrt{\sum_{\kappa=0}^{2^n-1}\abs{\sum_{j=0}^Q\sigma_jx^j_\kappa}^2}}\]
with probability $P\sim \mathcal{F}$ (see Definition~\ref{definition: filling ratio}) and resources no greater than:
\begin{enumerate}
    \item $\mathcal{O}(Qn\log_2n)$ quantum gates (C-NOTs and one-qubit rotations);
    \item $\lceil\log_2n\rceil+1$ pure ancillas.
\end{enumerate}
\end{theorem}

\begin{proof}
    The proof is provided in Appendix~\ref{appendix: cooridnate-polynomial-oracle}. The central idea underlying the Theorem is to apply the Singular Value Transformation to a diagonal unitary operator with linear eigenvalues which represent coordinate values on the interval $[-1,1]$.
\end{proof}

\begin{remark}
    While our method does not rely directly on properties such as smoothness or Lipschitz continuity, these characteristics do influence the efficiency of the polynomial approximation. For example, smooth or Lipschitz-continuous functions can often be approximated to a given precision $\epsilon$ using lower-degree polynomials, which in turn reduces the gate complexity of our circuit. On the other hand, functions with discontinuities or sharp features typically require higher-degree polynomials or a piece-wise representation.
\end{remark}

\begin{theorem}[Piece-wise polynomial function]\label{theorem:intro 2}
Let $f(x)$ be a piece-wise continuous function $f: \mathbb{R} \rightarrow \mathbb{C}$ with polynomial decomposition
\[
f(x)=\begin{cases} 
f_1(x) = \sum_{i=0}^{Q_1} \sigma_i^{(1)} x^i = \alpha_1^{(1)} P_1^{(1)}(x) + \alpha_2^{(1)} P_2^{(1)}(x) + i \alpha_3^{(1)} P_3^{(1)}(x) + i \alpha_4^{(1)} P_4^{(1)}(x), & \text{if } K_1 \geq x \geq a \\
f_2(x) = \sum_{i=0}^{Q_2} \sigma_i^{(2)} x^i = \alpha_1^{(2)} P_1^{(2)}(x) + \alpha_2^{(2)} P_2^{(2)}(x) + i \alpha_3^{(2)} P_3^{(2)}(x) + i \alpha_4^{(2)} P_4^{(2)}(x), & \text{if } K_2 \geq x > K_1 \\
\multicolumn{2}{c}{\vdots} \\
f_G(x) = \sum_{i=0}^{Q_G} \sigma_i^{(G)} x^i = \alpha_1^{(G)} P_1^{(G)}(x) + \alpha_2^{(G)} P_2^{(G)}(x) + i \alpha_3^{(G)} P_3^{(G)}(x) + i \alpha_4^{(G)} P_4^{(G)}(x), & \text{if } b \geq x > K_{G-1}
\end{cases};
\]
\[\sigma_i^{(j)}\in\mathbb{C};\quad \alpha_s^{(j)}\in\mathbb{R}\]
where $G$ - is the number of the pieces, and $Q_g$ - is the polynomial degree in the piece $g$ with $\max_g Q_g=Q_{\max}$. The polynomials $P^{(j)}_s(x)$ satisfy the conditions \ref{cond1}, \ref{cond2}, \ref{cond3}. Then we can construct a unitary operation $U_f$ that prepares a $2^n$-discrete version of $f(x)$ with probability $P\sim\mathcal{F}$ (see Definition~\ref{definition: filling ratio}) and resources no greater than:
\begin{enumerate}
    \item $\mathcal{O}(Q_{\max}n\log_2n+Gn+Q_{\max}G)$ C-NOTs and one-qubit operations;
    \item $n+\lceil \log_2G\rceil-1$ pure ancillas.
\end{enumerate}
\end{theorem}

\begin{proof}
    The proof is provided in Appendix~\ref{appendix:piece-wise}. The main idea is to prepare an auxiliary register in the state $\ket{1}$ for the interval \( K_{i} \geq x > K_{i-1} \), then apply the controlled version of the quantum circuit from Theorem~\ref{theorem: Linear combination of polynomials}.
\end{proof}

A subtle limitation of our approach arises when the filling ratio
$\mathcal{F}$ (Definition~\ref{definition: filling ratio}) becomes exponentially small. For typical continuous
or piecewise–continuous functions this does not occur, and the resource
scaling in Theorems~\ref{theorem:intro 1} and~\ref{theorem:intro 2} remains valid. However, for highly
localized functions whose weight concentrates on only a few grid
points, one may have $\mathcal{F}=\mathcal{O}(2^{-n})$, in which case the
postselection inherent in the block–encoding must be repeated
$1/\mathcal{F}=\mathcal{O}(2^n)$ times undermining the algorithm efficiency. These cases lie outside the main regime of interest
for this work, but we list two strategies how the filling ratio can be boosted by simple
preprocessing of the input amplitudes, exploiting the fact that our
framework yields with unitary which acts on the state $|+\rangle^{\otimes n}$. We outline representative strategies:

\begin{itemize}
\item \textbf{Sectoring.} When $f(x)$ contains a narrow peak near $s$,
one may replace the uniform input state $|+\rangle^{\otimes n}$ with a
two--level, piecewise--constant amplitude profile
$|g\rangle=\sum_{k} g(x_{k})\,|k\rangle$, where $g(x)$ takes a small value
$C_{1}$ outside $[s-\varepsilon,s+\varepsilon]$ and a larger value $C_{2}$
inside it. Such states correspond to a single rectangular ``feature'' and
can be prepared efficiently using the feature--based state--preparation
method~\cite{PhysRevA.73.012307}. Thus, the framework considered in our paper implements the
effective function $f'(x)=f(x)/g(x)$, which has a larger filling ratio,
though it requires a separate polynomial decomposition of $f'$, while
postselection still recovers the normalized state corresponding to the
original function $f(x)$.
\item \textbf{Approximate-profile compensation.} Suppose one can
efficiently prepare a simple function $g(x)$ approximating $f(x)$, and
has a polynomial decomposition of the ratio $f(x)/g(x)$. We first upload
$g$ as the input state and then apply our function-uploading procedure
to $f/g$, thereby transforming $|g\rangle$ into the desired state for
$f$. Since $g$ captures the dominant shape, the effective filling ratio
$\mathcal{F}$ of $f/g$ is significantly increased, including the case of exponentially
decaying functions. Sectoring corresponds to the special case where
$g$ is a two-level profile.
\end{itemize}

These techniques mitigate the exponentially small–$\mathcal{F}$ regime for certain cases, although
a full treatment is beyond the scope of this paper.

\section{Related Work}\label{section: related work}
There are numerous methods addressing the problem at hand.
Ref.~\cite{araujo2023configurable} proposes an algorithm
for uploading any vector (without an analytical expression)
with circuit depth and width of $\mathcal{O}(\sqrt{2^n})$,
which is crucial for fields like quantum machine learning
\cite{biamonte2017quantum}, where a quantum computer
operates on arbitrary states. Additionally, various
optimization approaches exist; one promising method
\cite{creevey2023gasp} uses a genetic algorithm to
iteratively build the desired quantum state. The authors
conjecture this method may scale as $\mathcal{O}(\mathrm{poly}(n))$.

Among other approaches, Ref.~\cite{araujo2023low} leverages
Schmidt decomposition. The complexity of this algorithm
depends on the Schmidt rank and can significantly reduce
the number of CNOT gates. The complexity is
$\mathcal{O}(2^m)$, where $m$ is the Schmidt measure.

Our approach not only prepares continuous functions that admit efficient
polynomial approximations and reasonably large filling ratio \(\mathcal{F}\),
but also provides an explicit method for constructing piece-wise polynomial
functions. This is achieved with
no change to the asymptotic scaling, offering a scalable
solution for approximating non-smooth or segmented data.
To the best of our knowledge, none of the below-mentioned
methods address this important instance explicitly, making
our method a uniquely flexible and practical alternative.

A particularly promising method is proposed in
Ref.~\cite{rosenkranz2024quantum}, where the authors
construct functions using a linear combination of
block-encodings for Fourier and Chebyshev series. Their
gate complexity is $\mathcal{O}(Qn\log n)$ for Chebyshev
series of degree $Q$, and $\mathcal{O}(dn\log d)$ for
Fourier series with harmonic $d$. The key difference
with our method lies in how polynomial superpositions
are formed. They employ LCU, while we rely on QSVT. Our method also requires fewer workspace ancillas: 
$\mathcal{O}(\log n)$ vs.~$\mathcal{O}(\log n + \log Q)$
or $\mathcal{O}(2\log d + 1)$, simplifying implementation.

Ref.~\cite{gonzalez2024efficient} is based on diagonal
block-encoding of amplitudes \cite{rattew2023non,
guo2021nonlinear} applied to a linearly-dependent circuit,
and then modified using singular value transformation.
While this algorithm is conceptually similar to ours,
it has worse complexity: $\mathcal{O}(Qn^2)$ versus
our $\mathcal{O}(Qn\log n)$. It also requires
$\mathcal{O}(n)$ workspace ancillas.

Another important early method is the Grover-Rudolph algorithm \cite{grover2002creating}, which prepares quantum states corresponding to integrals of a given probability density function. This method scales as $\mathcal{O}(poly(n))$ for certain smooth distributions and provides a strong foundation for amplitude encoding. Additionally, it does not depend on the filling ratio $\mathcal{F}$; its probability of success is $1$.

The main disadvantage of the Grover--Rudolph approach is its reliance
on quantum amplitude oracles that prepare conditional probabilities of
the form
\begin{equation}
f(i) = 
\frac{
    \displaystyle\int_{x_L^i}^{\frac{x_R^i + x_L^i}{2}}
    p(x)\,dx
}{
    \displaystyle\int_{x_L^i}^{x_R^i}
    p(x)\,dx
}
\end{equation}
where $x_L^i, x_R^i$ are subinterval endpoints and $p(x)$ is a target
probability density. This formulation assumes access to efficient
integration or a known cumulative distribution function (CDF), which
limits the method’s applicability. In each recursive step, the method
requires the transformation
\begin{equation}
\sqrt{p_i^{(m)}}\,|i\rangle\,|0\cdots 0\rangle 
\longrightarrow 
\sqrt{p_i^{(m)}}\,|i\rangle\,|\theta_i\rangle,
\qquad
\theta_i = \arccos\left( \sqrt{f(i)} \right)
\end{equation}
Now assume the target amplitude function is a polynomial
\begin{equation}
\sqrt{p(x)} = \sum_{k=0}^{Q} a_k\, x^k;\qquad
p(x) = \left( \sum_{k=0}^Q a_k\, x^k \right)^2
= \sum_{m=0}^{2Q} b_m\, x^m
\end{equation}

Then the conditional probability becomes
\begin{align}
f(i) &= 
\frac{\text{poly}_{2Q+1}}{\text{poly}_{2Q+1}};
\qquad
\theta_i = 
\arccos\left( \sqrt{
    \frac{
        \text{poly}_{2Q+1}
    }{
        \text{poly}_{2Q+1}
    }
} \right).
\end{align}

Thus, although we started with a seemingly simpler goal—quantum
state preparation for a degree $Q$ polynomial—the method ends
up requiring the construction of oracles for much more complicated
functions of degree $2Q+1$ and transcendental transformations. This significantly increases both classical and quantum implementation complexity. The authors do
not specify how to construct such oracles, which is nontrivial in
practice.


Earlier QSVT–based approach~\cite{mcardle2022quantum} requires approximating the composed function \(f(\arcsin(\cdot))\). This composition tends to make polynomial approximation harder: ensuring boundedness and parity on \([-1,1]\) is more delicate, coefficients can be less well-conditioned, and achieving the same approximation error may require a higher polynomial degree \(Q\).

Other lines of work emphasize regimes with small heralded success probabilities~\cite{li2023efficient} or concentrate on low-degree polynomial cases.





\section*{Acknowledgments}

N. Liu acknowledges funding from the Science and Technology Program of  
Shanghai, China (21JC1402900), NSFC grants  
No. 12471411 and No. 12341104, the Shanghai Jiao Tong University  
2030 Initiative, and the Fundamental Research Funds for the  
Central Universities. N. Guseynov acknowledges funding from NSFC grant W2442002.

\bibliography{references}

\appendix
\section{Explicit quantum circuits for the multi-control unitaries}\label{appendix: multi-control}

\begin{figure}[H]
\begin{centering}

    \subcaptionbox{}{\includegraphics[width=0.283\textwidth]{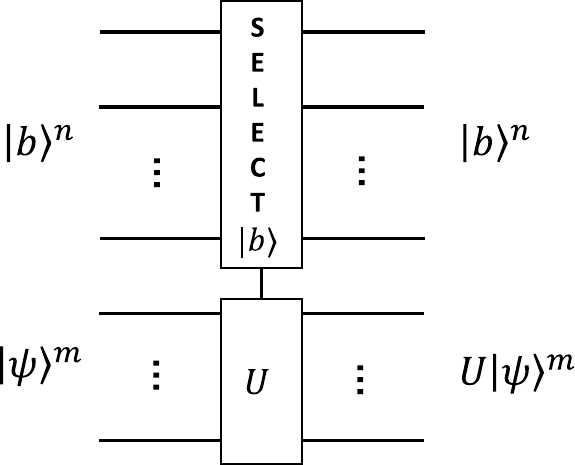}}
    \hspace{0.1\textwidth}
    \subcaptionbox{}{\includegraphics[width=0.48\textwidth]{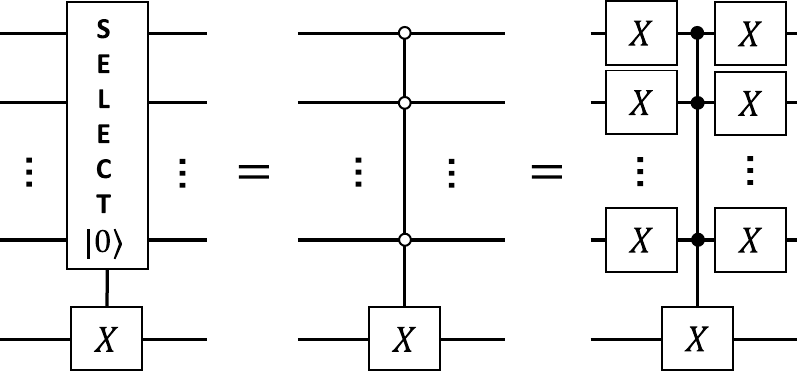}}
    \vspace{2.0em}
    \subcaptionbox{}{\includegraphics[width=0.68\textwidth]{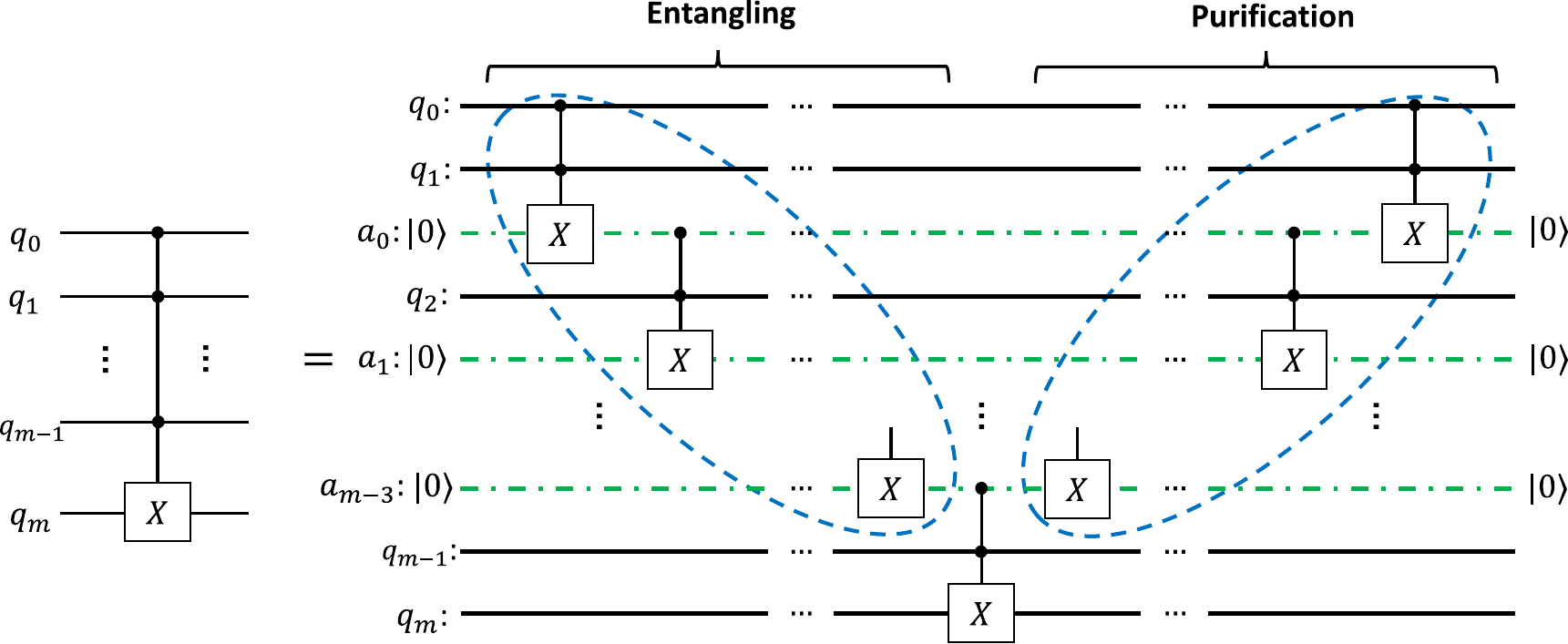}}

\end{centering}
    \caption{(a) The general view of multi-control operator $C_U^b$ from the definition \ref{def:multiconrol operator}. Wires for pure ancilla are depicted as green dash-dotted lines; and the system register is represented by a solid black wire. Select operation means that the operator $U$ is applied only if the state on the upper register is $\ket{b}^n$ (b) Explicit quantum circuit for implementation $C_X^{00\dots 0}$. (c) Explicit quantum circuit for implementation $C^{11\dots1}_X$ which we use everywhere in this paper. We use $2m-3$ Toffoli application and $m-2$ pure ancillas.}
    \label{fig:Multi_control_qc}
\end{figure}

In this section, we explicitly demonstrate how we address multi-control operations and gate counting. The general structure of a multi-control operator is illustrated in Fig.~\ref{fig:Multi_control_qc}. It is known \cite{nielsen2002quantum} that Toffoli gate can be realized using $6$ C-NOTs and $8$ single-qubit operations. Consequently, the resource requirement for a single Multi-control operator includes no more than one $C^1_U$ operator, $16n-16$ single-qubit operations, and $12n-12$ C-NOTs, simplifying the complexity of $C^{b}_X$ to $C^{11\dots 1}_X$ for simplicity. Additionally, we utilize $n-1$ pure ancillas. In this paper, we compute the complexity of the controlled unitary $C^1_U$ based on the gate complexity of $U$ by substituting all C-NOT gates with Toffoli gates (each implemented using $6$ C-NOTs and $8$ single-qubit operations), and by replacing single-qubit gates with their controlled versions (requiring $2$ C-NOTs and $2$ single-qubit rotations), following the standard construction in \cite{nielsen2002quantum}.

\section{Efficient Series Expansions and Error Terms}

In this Appendix, we provide an overview of efficient series expansions for common functions and their corresponding error terms. Table~\ref{table:1} presents the notation, series expansions, error expressions, and domains for a variety of functions such as exponent, trigonometric functions, and others. Below, we define the terms and explain how the error is computed.

\begin{itemize}
    \item Series Expansion: Each function is approximated by truncating its series expansion up to $k$ terms. For example, the expansion for the exponential function $e^x$ is $\sum_{s=0}^k\frac{x^s}{s!}$. These expansions provide approximate values for the functions within a specific domain.
    \item Error Expression: The error term $\epsilon$ represents the residual sum of the terms beyond $k$ in the series expansion. For each function, the error provides a measure of the approximation accuracy.
    \item Norm of the Error: The error terms are presented in the infinity norm, denoted as $\|\cdot\|_{\infty}$, which corresponds to the maximum absolute value of the error across the given domain.
\end{itemize}

\begin{table}[H]
\centering
\begin{tabular}{|c|c|c|c|c|}
\hline
\textbf{Name} & \textbf{Notation} & \textbf{Series Expansion} & \textbf{Error Expression} & \textbf{Domain} \\
\hline
Exponent & $e^{\alpha x}$ & $\displaystyle\sum_{s=0}^{k} \frac{(\alpha x)^s}{s!}$ & $\displaystyle\epsilon \approx\frac{e^{|\alpha|+k+1}\, |\alpha|^{k+1}}{\sqrt{2\pi (k+1)}\, (k+1)^{k+1}}$ & $x \in [-1, 1]; \alpha \in \mathbb{C}$ \\
\hline
Cosine & $\cos(tx)$ & $J_0(t) + 2 \sum_{s=1}^{k} (-1)^s J_{2s}(t) T_{2s}(x)$ & $\epsilon \leq \frac{1.07}{\sqrt{2(k+1)}} \left( \frac{e|t|}{4(k+1)} \right)^{2(k+1)}$ & $x \in [-1, 1]; t \in \mathbb{R}$ \\
\hline
Sine & $\sin(tx)$ & $2 \sum_{s=0}^{k} (-1)^s J_{2s+1}(t) T_{2s+1}(x)$ & $\epsilon \leq \frac{1.07}{\sqrt{2(k+1)}} \left( \frac{e|t|}{4(k+1)} \right)^{2(k+1)}$ & $x \in [-1, 1]; t \in \mathbb{R}$ \\
\hline
Gaussian & 
$\frac{1}{\sigma\sqrt{2\pi}}\, e^{-\frac{x^2}{2\sigma^2}}$ & 
$\frac{1}{\sigma\sqrt{2\pi}}\, \displaystyle\sum_{s=0}^{k} \frac{(-1)^s}{s!} \left( \frac{x^2}{2\sigma^2} \right)^s$ & 
$\displaystyle\epsilon \approx \frac{1}{\sigma 2\pi \sqrt{k+1}}
\left(
\frac{e}{2(k+1)\sigma^2}
\right)^{k+1}$ & 
$x \in [-1, 1]; \sigma > 0$ \\
\hline
Sigmoid & $\sigma(x) = \frac{1}{1 + e^{-x}}$ & $\frac{1}{2} + \frac{1}{2} \sum_{s=1}^{k} \frac{2^{2s}(2^{2s}-1)B_{2s}}{(2s)!} \left( \frac{x}{2} \right)^{2s-1}$ & $\epsilon \approx \left| \frac{2^{2(k+1)}(2^{2(k+1)}-1)B_{2(k+1)}}{(2(k+1))!} \right|$ & $x \in [-\pi/2, \pi/2]$ \\
\hline
Reciprocal & $\frac{1}{x}$ & $4 \sum_{j=0}^{k} (-1)^j \sum_{s=j+1}^{K} \frac{\binom{2K}{K+s}}{2^{2K}} T_{2j+1}(x)$ & $\epsilon \leq \exp\left( -\frac{-\log(\delta) + \sqrt{\log(\delta)^2 + 4 \left( \frac{k}{\delta} \right)^2}}{2} \right)$ & $x \in [-1, -\frac{1}{\delta}] \cup [\frac{1}{\delta}, 1]$ \\
\hline
Bessel & $J_n(\alpha x)$ & $\sum_{s=0}^{k} \frac{(-1)^s}{s! \, \Gamma(s+n+1)} \left(\frac{\alpha x}{2}\right)^{2s+n}$ & $\epsilon \approx \frac{\left(\frac{|\alpha|}{2}\right)^{2(k+1)+n} e^{k + n + 1}}{2\pi \, (k+1)^{k + \frac{3}{2}} (k+n+1)^{k + n + 2}}$ & $x \in [-1, 1]; \alpha \in \mathbb{R}$ \\
\hline
Hyp. Tan. & $\tanh(x)$ & $\sum_{s=1}^{k} \frac{2^{2s} (2^{2s} - 1) B_{2s} x^{2s-1}}{(2s)!}$ & $\epsilon \approx \left| \frac{2^{2(k+1)} (2^{2(k+1)} - 1) B_{2(k+1)} x^{2(k+1)-1}}{(2(k+1))!} \right|$ & $x \in [-\pi/2, \pi/2]$ \\
\hline
\end{tabular}
\caption{Efficient Series Expansions and Error Terms \cite{childs2017quantum,abramowitz1968handbook}; as the error expression we use the absolute value of the number of terms retained $\abs{\sum_{k+1}^\infty\dots}$. 
  $J_n(x)$: Bessel function of the first kind; 
  $T_n(x)$: Chebyshev polynomials of the first kind; 
  $B_n$: Bernoulli numbers; 
  $\Gamma(x)$: Gamma function. 
  We approximate each function using its series expansion up to $k$ terms. 
  For the reciprocal function, $K = \left\lceil \delta^2 \log(\delta \epsilon^{-1}) \right\rceil$.
}
\label{table:1}
\end{table}

\section{Coordinate-polynomial-oracle}\label{appendix: cooridnate-polynomial-oracle}

The primary objective of this paper is to construct a function 
defined on the interval $[-1,1]$. We begin by constructing an amplitude oracle for the coordinate operator $\hat{x}$ on the interval $[-1,1]$. The operator $\hat{O}_x$ acts as follows:
\begin{eqnarray}
    \hat{O}_x\ket{0}^\lambda\ket{k}^n=\frac{1}{\mathcal{N}_x}x_k\ket{0}^\lambda\ket{k}^n+J_{x}\ket{\bot_0}^{n+\lambda};\quad x_k=-1+\frac{2k}{2^n-1},
    \label{eq: x Coordinate oracle}
\end{eqnarray}
where $\ket{\bot_0}^{n+\lambda}$ denotes a state orthogonal 
to $\ket{0}^\lambda \otimes I^{\otimes n}$ meaning $\forall\ket{\psi}^n:\bra{\bot_0}^{n+\lambda}\left(\ket{0}^\lambda\otimes\ket{\psi}^n\right)=0$. The second term involving $J_x$ corresponds to the orthogonal component introduced by the block-encoding structure and ensures unitarity of the oracle $\hat{O}_x$.
While this term can, in principle, be explicitly reconstructed from the following quantum circuits, we omit its full expression here as it is not important for the understanding of the algorithm and would be unnecessarily cumbersome.

Thus, $\hat{O}_x$ 
provides a $(\mathcal{N}_x, \lambda, 0)$-block-encoding of the 
coordinate operator $\hat{x}$. In this paper we construct this operator using the LCU technique:
\begin{eqnarray}
\hat{x} = \begin{pmatrix}
-1 &  &  \\
 & -1+\frac{2}{2^n-1} &  \\
 &  & \ddots \\
 &  &  & 1
\end{pmatrix}=-\sum_{j=0}^{n-1}\frac{2^j}{2^n-1}Z_j;\quad Z_j=I\otimes \dots \otimes I\otimes Z\otimes \underbrace{I\otimes \dots \otimes I}_{j \text{ times}},
\label{eq: LCU for x}
\end{eqnarray}
where $Z=\begin{pmatrix}
    1&0\\
    0&-1
\end{pmatrix}$ is the Pauli operator.

\begin{lemma}[Amplitude-oracle for coordinate operator (Appendix A from \cite{guseynov2023depth})]\label{lemma coordinate oracle}
Let $2^n\times 2^n$ matrix $\hat{x}=diag(-1,-1+\frac{2}{2^n-1},\dots,1)$ be a $n$-qubit finite-difference representation of coordinates as in the Eq.~(\ref{eq: LCU for x}), then we can construct a $(1,\lceil \log_2n\rceil,0)$-block-encoding of the matrix $\hat{x}$:
\[
 \hat{O}_x\ket{0}^\lambda\ket{k}^n=x_k\ket{0}^\lambda\ket{k}^n+J_{x}\ket{\bot_0}^{n+\lambda};\quad x_k=-1+\frac{2k}{2^n-1};\quad \lambda=\lceil\log_2n\rceil,
\]
with resources no greater than:
\begin{enumerate}
    \item $16n\lceil \log_2n\rceil-16n$ one-qubit operations;
    \item $12n\lceil \log_2n\rceil-16n$ C-NOTs;
    \item $\lceil \log_2n\rceil-2$ pure ancillas.
\end{enumerate}
\end{lemma}

\begin{proof}
We prove this lemma by demonstrating the circuit in Fig.~\ref{fig: O_x Amplitude oracle for x}, which implements the 
LCU technique to build the desired operator.
\begin{figure}[H]
    \centering
        \includegraphics[width=0.7\textwidth]{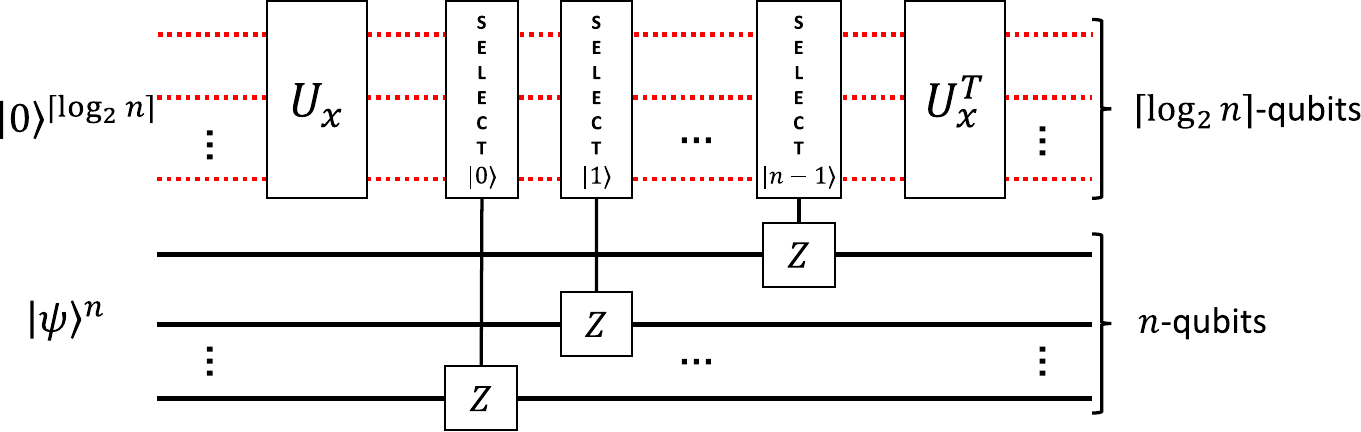}
    \caption{A construction of the Amplitude-oracle for $2^n\times 2^n$ matrix $\hat{x}$ which is $(1,\lceil \log_2n\rceil,0)$-block-encoding. The construction using LCU technique as in Eq.~(\ref{eq: LCU for x}). The unitary $U_x$ is defined in the Eq.~(\ref{eq:aux superposition for O_x}) with its transposed version $U_x^T$. Workspace ancilla are shown as red dotted wires; and the system register is represented by a solid black wire.}
    \label{fig: O_x Amplitude oracle for x}
\end{figure}

The implementation utilizes the auxiliary unitary:
    \begin{eqnarray}
        U_x=\frac{i}{\sqrt{{2^{n}-1}}}\sum_{j=0}^{n-1}2^{j/2}\ket{j}^\lambda
        \label{eq:aux superposition for O_x}
    \end{eqnarray}
    The construction of $U_x$ and $U_x^T$ addresses the general problem of construction an arbitrary quantum state. The authors of \cite{plesch2011quantum,bergholm2005quantum} suggest an explicit implementation of $U_x$ which consist of less than $\frac{23}{24}n$ C-NOTs for even number of qubits and $n$ C-NOTs for the odd number. Let's consider the number of one-qubit operations here. Naively, one can conclude that each C-NOT can be accompanied by $6$ Euler rotations; however, a rotation around $z$ axis commutes with control operation, and a rotation around $x$ one commutes with Pauli $X$ operation. Thus, the number of one-qubit Euler rotations cannot exceed the number of C-NOTs by more than 4 times \cite{shende2004minimal}. For simplicity in this paper we address a general one-qubit operation, so we merge 2 Euler rotations acting on each two qubits where the C-NOT is applied. Thus, the complexity of the $U_\alpha^H$ implementation is less than: (i) $2n$ one-qubit rotations, (ii) $n$ C-NOTs.
\end{proof}

\subsection{From $x$ to a polynomial}

In this subsection we use the trick called qubitization \cite{gilyen2019quantum} to transform the Amplitude-oracle for the $2^n\times 2^n$ matrix $\hat{x}$ denoted as $\hat{O}_x$ to the Amplitude-oracle for a polynomial of $\hat{x}$ to $\hat{O}_{P_s}$ (Coordinate-polynomial-oracle). The main building block in this trick is so called Alternating phase modulation sequence. This unitary operation is a key instrument for building $\hat{O}_{P_s}$.

\begin{definition}[Alternating phase modulation sequence]
    \label{def:Alternating phase modulation sequence}
    Let $U$ be a $(1,m,0)$-block-encoding of hermitian matrix $A$ such that
    \[A=(\ket{0}^m\bra{0}^m\otimes I^{\otimes n})U(I^{\otimes n}\otimes\ket{0}^m\bra{0}^m);\]
    let $\Phi\in\mathbb{R}^q$, then we define the $q$-phased alternating sequence $U_\Phi$ as follows
    \[
U_{\Phi} := 
\begin{cases} 
e^{i \phi_1 (2 \Pi - I)} U \prod_{j=1}^{(q-1)/2} \left( e^{i \phi_{2j} (2 \Pi - I)} U^\dagger e^{i \phi_{2j+1} (2 \Pi - I)} U \right) & \text{if } q \text{ is odd, and} \\
\prod_{j=1}^{q/2} \left( e^{i \phi_{2j-1} (2 \Pi - I)} U^\dagger e^{i \phi_{2j} (2 \Pi - I)} U \right) & \text{if } q \text{ is even,}
\end{cases}
\]
where $2\Pi-I=(2\ket{0}^m\bra{0}^m-I^{\otimes m})\otimes I^{\otimes n}$.
\end{definition}

Now we apply the Alternating phase modulation sequence on $\hat{O}_x$. In this case $m=1$, $e^{i \phi (2 \Pi - I)}$ turns to implementing the one-qubit rotation $C^{0\dots0}_{\Phi(-2\phi)}$
\begin{equation}
    \Phi(\theta) = \begin{pmatrix}
    e^{i\theta} & 0 \\
    0 & 1
    \end{pmatrix}.
\label{eq:Rz}
\end{equation}
The quantum circuit for $U_\Phi$ for this case is depicted in Fig.~\ref{fig:U_phi}.

\begin{figure}[H]
\centering
    \includegraphics[width=0.7\textwidth]{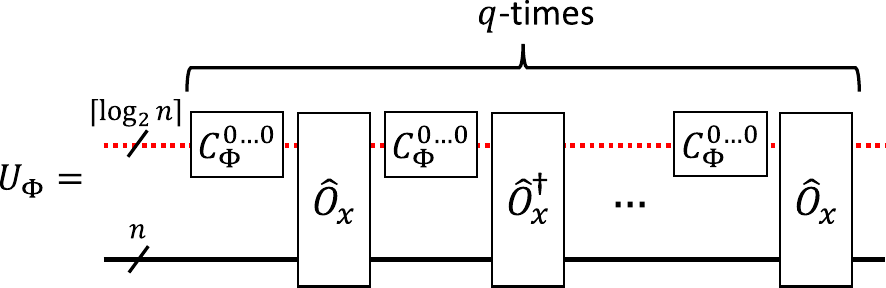}
    \caption{Alternating phase modulation sequence $U_\Phi$ for $\hat{O}_x$ with odd $q$. Workspace ancilla are shown as red dotted wires; and the system register is represented by a solid black wire.}
    \label{fig:U_phi}
    \end{figure}

Now we have all the ingredients to build the Coordinate-polynomial-oracle $\hat{O}_{P_s}$ using singular value transformation.

\begin{theorem}[Singular value transformation by real polynomials (Corollary 18 from \cite{gilyen2019quantum})]\label{theorem:Singular value transformation by real polynomials}
    Let $P_{q}(x)$ be a degree-$Q$ polynomial as in Eq.~(\ref{eq:four_polyn}) and $U_\Phi$ is the Alternating phase modulation sequence as in Definition~\ref{def:Alternating phase modulation sequence} based on Amplitude-oracle for $n$-qubit matrix $\hat{x}=diag(-1,-1+\frac{2}{2^n-1},\dots,1)$ as in Lemma~\ref{lemma coordinate oracle}. Then there exist $\Phi\in\mathbb{R}^n$, such that
    \[
    \hat{O}_{P_s} = \left(H_W\otimes I^{\otimes n+\lceil\log_2n\rceil+1} \right) \left( |0\rangle \langle 0| \otimes U_{\Phi} + |1\rangle \langle 1| \otimes U_{-\Phi} \right)\left(H_W\otimes I^{\otimes n+\lceil\log_2n\rceil+1} \right),
    \]
    where $H_W = \frac{1}{\sqrt{2}}
\begin{pmatrix}
1 & 1 \\
1 & -1
\end{pmatrix}$. $\hat{O}_{P_s}$ is Coordinate-polynomial-oracle which is a $(1,{\lceil\log_2n\rceil+1},0)$-block-encoding of $P_s(\hat{x})$:
    \begin{eqnarray}
    \hat{O}_{P_s}\ket{0}^{\lceil\log_2n\rceil+1}\ket{i}^n=(P_s(x))_{i}\ket{0}^{\lceil\log_2n\rceil+1}\ket{i}^n +J^{(i)}_{P_s}\ket{\bot_0}^{n+{\lceil\log_2n\rceil+1}}.
    \label{eq:O_Ps}
    \end{eqnarray}
    The total complexity of construction $\hat{O}_{P_s}$ does not exceed:
    \begin{enumerate}
    \item $16Qn\lceil \log_2n\rceil+8Qn+32Q\lceil \log_2n\rceil-48q+2$ one-qubit operations;
    \item $12Qn\lceil \log_2n\rceil+4Qn+24Q\lceil \log_2n\rceil-36q$ C-NOTs;
    \item $\lceil \log_2n\rceil-1$ pure ancillas.
    \end{enumerate}
    Moreover, given $P_s(x)$ and $\delta \geq 0$ we can find a $P^\prime_s(x)$ and a corresponding $\Phi$, such that $|P^\prime_s(x) - P_s(x)| \leq \delta$ for all $x \in [-1,1]$, using a classical computer in time $\mathcal{O}(\text{poly}(Q, \log(1/\delta)))$.
\end{theorem}
\begin{proof}
    The rule of choosing $\Phi$ and the general proof is given in \cite{gilyen2019quantum}. We computed complexities using Amplitude-oracle for $x$ (Lemma~\ref{lemma coordinate oracle}) as an input for $U_\Phi$ depicted in Fig.~\ref{fig:U_phi}.
\end{proof}

\begin{remark}\label{remark: O_x no need control}
    To estimate the gate complexity of the Coordinate-polynomial-oracle, we note that the LCU implementing the controlled unitaries $U_\Phi$ and $U_{-\Phi}$ admits a useful simplification. As illustrated in Fig.~\ref{fig:U_phi}, the only distinction between $\ketbra{0}{0}\!\otimes\! U_\Phi$ and $\ketbra{0}{0}\!\otimes\! U_{-\Phi}$ lies in the phase operations $C^{00\cdots 0}_\Phi$; the Amplitude-oracle $\hat{O}_x$ is identical in both branches. Hence the calls to $\hat{O}_x$ can be shared and need not be controlled by the selector qubit (whether it is in $\ket{0}$ or $\ket{1}$). Dropping these redundant controls on $\hat{O}_x$ reduces the overall LCU overhead and yields the gate counts stated above.
\end{remark}

Our construction relies on precomputed QSP phase angles for the
polynomials in Theorem~\ref{theorem:Singular value transformation by real polynomials}. We do not design new algorithms for this step,
but instead assume access to existing classical solvers for quantum signal
processing. Optimization–based methods such as \cite{QSP_angles1} and stable factorization schemes
such as Ying~\cite{QSP_angles2} provide numerically robust phase
synthesis in double precision and have been demonstrated for polynomial
degrees up to $10^{4}$. Open–source implementations, for example
the Python library \texttt{pyqsp}~\cite{pyqsp}, can be used as the
classical preprocessing backend for our scheme. A detailed analysis of the
classical cost of phase computation is beyond the scope of this work, and
we treat it as an assumed primitive.

Now we use LCU to combine four coordinate-polynomial-oracles to build the discrete version of $f(x)$ from Eq.~(\ref{eq:wave-function representation_discrete from a to b}).

\begin{theorem}[Linear combination of polynomials \{Same as Theorem~\ref{theorem:intro 1}\}]\label{theorem: Linear combination of polynomials}
Let $f(x)$ be a continuous function $f: \mathbb{R} \rightarrow \mathbb{C}$ with polynomial decomposition
\[f(x)=\sum_{i=0}^Q\sigma_ix^i=\alpha_1P_1(x)+\alpha_2P_2(x)+i\alpha_3P_3(x)+i\alpha_4P_4(x);\qquad \sigma_i\in\mathbb{C};\quad \alpha_s\in\mathbb{R},\]
where polynomials $P_s(x)$ satisfy:
\begin{itemize}
    \item $P_s(x)$ has parity-($s$ mod $2$);
    \item $P_s(x)\in\mathbb{R}[x]$;
    \item for all $x\in[-1,1]$: $\abs{P_s(x)}<1$.
\end{itemize}

Then we can construct a unitary operation $U_f$ which is $(\sum_{\kappa=0}^4\abs{\alpha_s},\lceil\log_2n\rceil+3,0)$ - block encoding of $f(\hat{x})$. $U_f$ prepares a $2^n$-discrete version of $f(x)$:
\[\ket{\psi}^n=\frac{\sum_{i=0}^Q\sum_{k=0}^{2^n-1}\sigma_ix^i_k\ket{k}^n}{\sqrt{\sum_{\kappa=0}^{2^n-1}\abs{\sum_{j=0}^Q\sigma_jx^j_\kappa}^2}}\]
with probability $P\sim\mathcal{F}$ (see Definition~\ref{definition: filling ratio}) and resources no greater than:
\begin{enumerate}
    \item $16Qn\lceil\log_2n\rceil-16Qn+32Q\lceil\log_2n\rceil+16n\lceil\log_2n\rceil+33n+120Q+2$ one-qubit operations;
    \item $12Qn\lceil\log_2n\rceil-16Qn+24Q\lceil\log_2n\rceil+12n\lceil\log_2n\rceil+26n+92Q$ C-NOTs;
    \item $\lceil\log_2n\rceil+1$ pure ancillas.
\end{enumerate}
\end{theorem}

\begin{proof}
We demonstrate the circuit design in Fig.~\ref{fig:linear_combination_of_polys}, where
\begin{equation}
    U_\alpha\ket{0}^2=\frac{1}{\sqrt{\sum_{s=1}^4\abs{\alpha_s}}}\left[\sqrt{\alpha_1}\ket{00}^2+\sqrt{\alpha_2}\ket{01}^2+e^{i\pi/4}\sqrt{\alpha_3}\ket{10}^2+e^{i\pi/4}\sqrt{\alpha_4}\ket{11}^2\right]
    \label{eq:linear_combination_of_pokys_clarification}.
\end{equation}
The operation $U_\alpha$ is an arbitrary two-qubit operation (3 C-NOTs, 6 one-qubit operations); the complexity of $\hat{O}_x$ is derived in Lemma~\ref{lemma coordinate oracle}. We highlight that during the creation of Coordinate-polynomial-oracles (see Theorem~\ref{theorem:Singular value transformation by real polynomials}) for all four polynomials it is necessary to apply $\hat{O}_x$ $Q-1$ times, so there is no need with their control versions; we apply the controlled version of $\hat{O}_x$ only once due to different parities, which is similar to Remark~\ref{remark: O_x no need control}.

The filling ratio $\mathcal{F}$ is an important quantity that
measures how the total probability amplitude is distributed
after applying the block-encodings $\hat{O}_{P_s}$ to the constant
superposition $H_W^{\otimes n}\ket{0}^n = \ket{+}^n$, as shown
in Fig.~\ref{fig:linear_combination_of_polys}, where $H_W = \frac{1}{\sqrt{2}}
\begin{pmatrix}
1 & 1 \\
1 & -1
\end{pmatrix}$. The block-encodings
$\hat{O}_{P_s}$ have matrix elements bounded by $1$, meaning that
after their action, the total probability tends to decrease;  the Fig.~\ref{fig:filling ratio} demonstrated how the applying of the $\hat{O}_{P_s}$ cut the rectangle leaving the blue area under the plot. The total probability $P\sim\mathcal{F}$ is derived from the first term in the Eq.~(\ref{eq:O_Ps}).

\begin{figure}[H]
\begin{centering}

    \includegraphics[width=0.68\textwidth]{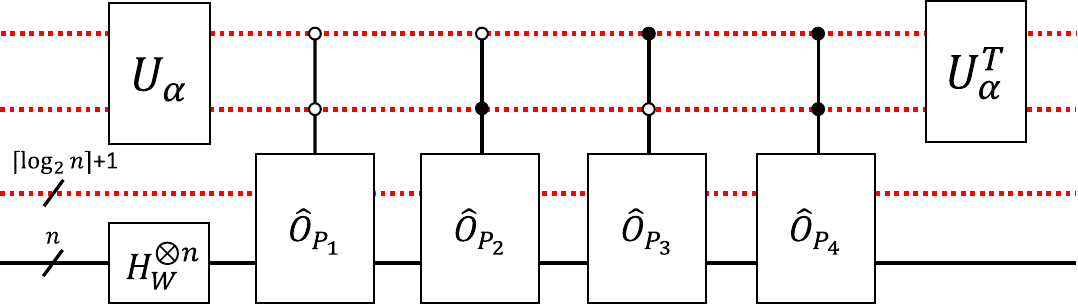}

\end{centering}
    \caption{Quantum circuit for the $2^n$-discrete version of $f(x)$. The desired function is prepared if the first three registers from the top are measured in the zero-state $\ket{0}$. The gate $U_\alpha$ is clarified in Eq.~(\ref{eq:linear_combination_of_pokys_clarification}) with a transposed version $U_\alpha^T$; the $H_W$ is Hadamard gate $H_W = \frac{1}{\sqrt{2}}
\begin{pmatrix}
1 & 1 \\
1 & -1
\end{pmatrix}.$ Workspace ancilla are shown as red dotted wires; and the system register is represented by a solid black wire.}
    \label{fig:linear_combination_of_polys}
\end{figure}
\end{proof}

\section{Piece-wise polynomial function}\label{appendix:piece-wise}

Now we generalize the pipeline described in Appendix~\ref{appendix: cooridnate-polynomial-oracle} for the more general case of non-continuous functions. We start with the simplest case of two pieces
\begin{equation}
    f(x)=\begin{cases}
\sum_{i=0}^{Q_1} \sigma_i^{(1)} x^i, & \text{if } a \leq x \leq K \\
\sum_{i=0}^{Q_2} \sigma_i^{(2)} x^i, & \text{if } b \geq x > K
\end{cases};\qquad \sigma_i^{(j)}\in\mathbb{C}
\label{eq:piece-wise}
\end{equation}
As in the previous section, our aim is to build a $2^n$-discrete version of $f(x)$ as a quantum state on a digital- or qubit-based quantum computer~\cite{FEYNMAN}.

We aim to reduce the problem to the previous case; for this purpose, we introduce one ancillary qubit which keeps the information whether $k\geq K_n$ or not, where $k$ is the number of a state in the computational basis and $K_n=\lceil(K-a)/\Delta x\rceil$ is a discrete coordinate of $K$. The initialization of this qubit (indicator qubit) is depicted in Fig.~\ref{fig:K_greater_or_not} and clarified in Fig.~\ref{fig:OR}. The complexity of this procedure is no greater than 
\begin{enumerate}
    \item $16n+34$ one-qubit operations,
    \item $12n-4$ C-NOTs,
    \item $n-1$ pure ancillas.
\end{enumerate}

\begin{figure}[h]
    \includegraphics[width=0.5\textwidth]{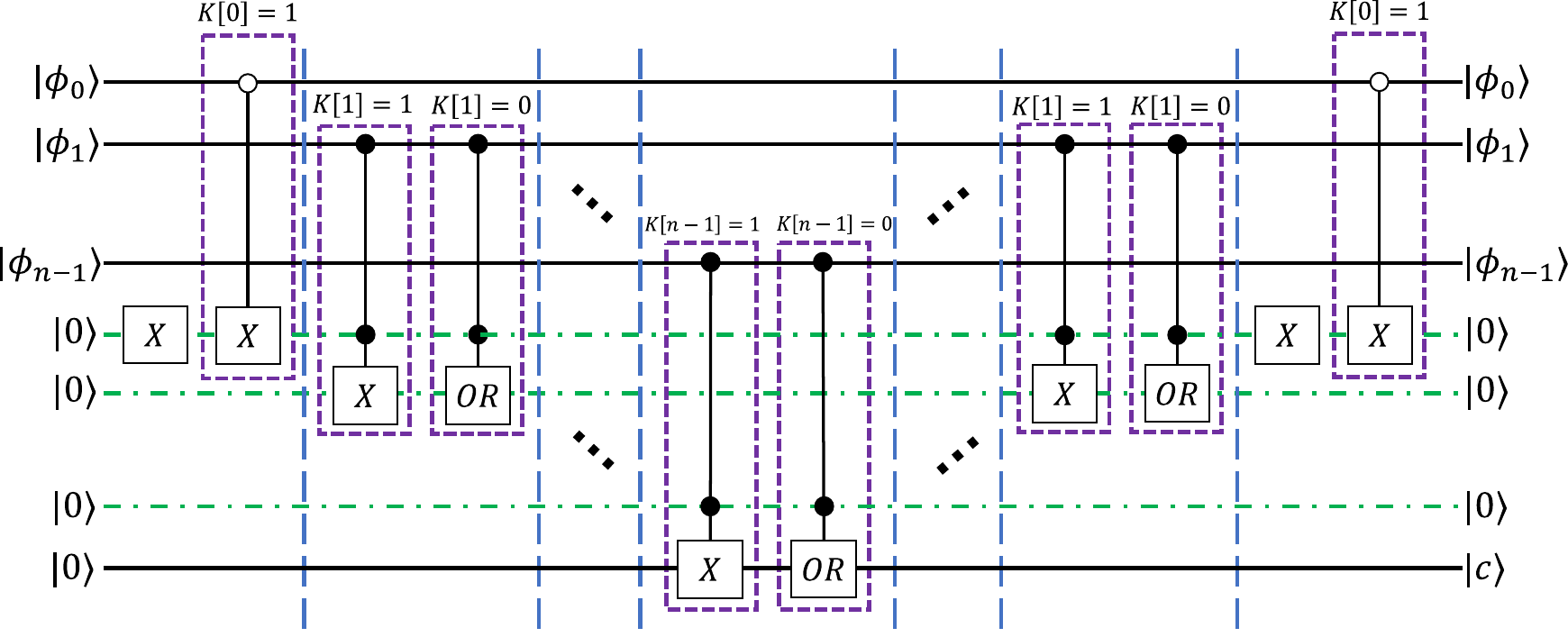}
    \caption{Circuit design for initializing the indicator qubit. The whole circuit can be understood as a classical comparator as it compares the state (computational basis) in the upper register with $K$, setting the last qubit in the state $\ket{c}=\ket{1}$ if $\phi\geq K_n$. The quantum circuit uses $n-1$ pure ancillas. A classical array $K[i]$ holds the binary representation of $K_n$; the value of $K[i]$ determines which gate is applied. Wires for pure ancilla are depicted as green dash-dotted lines; and the system register is represented by a solid black wire.}
    \label{fig:K_greater_or_not}
\end{figure}

\begin{figure}[h]
    \includegraphics[width=0.4\textwidth]{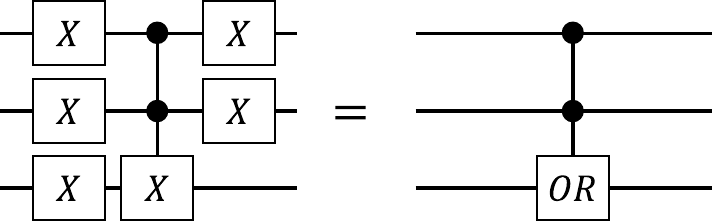}
    \caption{The representation of the $OR$ gate used in Fig.~\ref{fig:K_greater_or_not}.}
    \label{fig:OR}
\end{figure}

Using this auxiliary feature, we can prepare the $2^n$-discrete version of $f(x)$ from Eq.~(\ref{eq:piece-wise}). The corresponding circuit is depicted in Fig.~\ref{fig:piece_wise_example}. By induction we prove Theorem~\ref{theorem: last theorem}.

\begin{figure}[t]
    \includegraphics[width=0.5\textwidth]{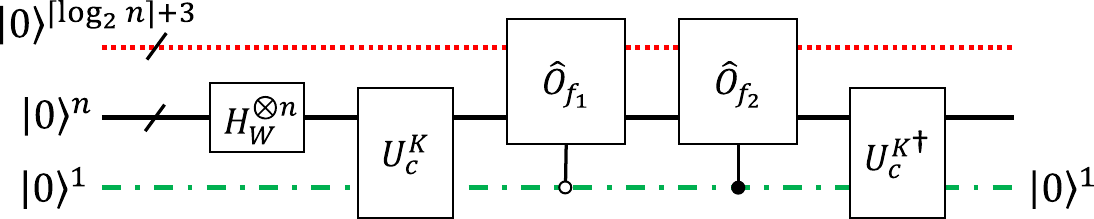}
    \caption{Circuit design for preparing a piece-wise function (2 pieces) from Eq.~(\ref{eq:piece-wise}). Here the unitary $U_c^K$ sets the last qubit in state $\ket{1}$ if the state in the middle $n$-qubit register is greater than $K_n$; the circuit design is depicted in Fig.~\ref{fig:K_greater_or_not}. The oracles $\hat{O}_{f_j}$ are a controlled version of quantum circuits from Theorem~\ref{theorem: Linear combination of polynomials} but without Hadamard gates $H_W^{\otimes n}$ (see Fig.~\ref{fig:linear_combination_of_polys}). The unitary $U_c^{K\dagger}$ returns the ancilla back to the zero state, making it a pure ancilla (see Definition~\ref{def:pure ancilla}). Wires for pure ancilla are depicted as green dash-dotted lines; workspace ancilla are shown as red dotted wires; and the system register is represented by a solid black wire.}
    \label{fig:piece_wise_example}
\end{figure}

\begin{theorem}[Piece-wise polynomial function \{Same as Theorem~\ref{theorem:intro 2}\}]\label{theorem: last theorem}
Let $f(x)$ be a piece-wise continuous function $f: \mathbb{R} \rightarrow \mathbb{C}$ with polynomial decomposition
\[
f(x)=\begin{cases} 
f_1(x) = \sum_{i=0}^{Q_1} \sigma_i^{(1)} x^i = \alpha_1^{(1)} P_1^{(1)}(x) + \alpha_2^{(1)} P_2^{(1)}(x) + i \alpha_3^{(1)} P_3^{(1)}(x) + i \alpha_4^{(1)} P_4^{(1)}(x), & \text{if } a \leq x \leq K_1 \\
f_2(x) = \sum_{i=0}^{Q_2} \sigma_i^{(2)} x^i = \alpha_1^{(2)} P_1^{(2)}(x) + \alpha_2^{(2)} P_2^{(2)}(x) + i \alpha_3^{(2)} P_3^{(2)}(x) + i \alpha_4^{(2)} P_4^{(2)}(x), & \text{if } K_2 \geq x > K_1 \\
\multicolumn{2}{c}{\vdots} \\
f_G(x) = \sum_{i=0}^{Q_G} \sigma_i^{(G)} x^i = \alpha_1^{(G)} P_1^{(G)}(x) + \alpha_2^{(G)} P_2^{(G)}(x) + i \alpha_3^{(G)} P_3^{(G)}(x) + i \alpha_4^{(G)} P_4^{(G)}(x), & \text{if } b \geq x > K_{G-1}
\end{cases};
\]
\[\sigma_i^{(j)}\in\mathbb{C};\quad \alpha_s^{(j)}\in\mathbb{R},\]
where polynomials $P^{(j)}_s(x)$ satisfy:
\begin{itemize}
    \item $P^{(j)}_s(x)$ has parity-($s$ mod 2);
    \item $P^{(j)}_s(x)\in\mathbb{R}[x]$;
    \item For all $x\in[-1,1]$: $\abs{P^{(j)}_s}<1$.
\end{itemize}

Then we can construct a unitary operation $U_f$ which is a block-encoding of $f(\hat{x})$ that prepares a $2^n$-discrete version of $f(x)$ with probability $P\sim\mathcal{F}$ and resources no greater than:
\begin{enumerate}
    \item $\mathcal{O}(Q_{\max}n\log_2n+Gn+Q_{\max}G)$ quantum gates;
    \item $n+\lceil \log_2 G\rceil-1$ pure ancillas.
\end{enumerate}

The number of workspace ancillas is $\lceil\log_2n\rceil+3$; and $Q_{\max}=\max_gQ_g$.
\end{theorem}

\begin{remark}[Why the costs split as $Q\,n\log n\;+\;G n\;+\;Q_{\max}G$]
\label{rem:reuse-x-oracle}
A key point is that all piece oracles 
$\hat{O}_{f_1},\hat{O}_{f_2},\dots,\hat{O}_{f_G}$ 
\emph{reuse the same} Amplitude-oracle for the coordinate operator $\hat{O}_x$. 
Within each $\hat{O}_{f_i}$ (see Fig.~\ref{fig:piece_wise_example}), the calls to $\hat{O}_x$ appear as a fixed subsequence that does \emph{not} depend on the state of the final indicator qubit. 
Hence these $\hat{O}_x$ invocations need not be controlled by the indicator; only their \emph{presence} and ordering are fixed per piece.

This observation explains the decomposition of our gate complexity:
\begin{enumerate}
  \item \textbf{$Q_{\max}\,n\log n$ term.} Because the same $\hat{O}_x$ sequence is reused across all pieces, the total cost to implement $G$ alternating phase modulation sequences (Def.~\ref{def:Alternating phase modulation sequence}) is dominated by repeatedly applying the longest \emph{uncontrolled} coordinate oracle sequence, yielding the $Q_{max}\,n\log n$ factor.

  \item \textbf{$G n$ term.} For piecewise application we execute the segment selector $U^{K}_{c}$ once per segment to prepare/unprepare the indicator (see Fig.~\ref{fig:piece_wise_example}), incurring an additional linear overhead in the number of segments, i.e., $G n$.

  \item \textbf{$Q_{\max}G$ term.} The only place where the indicator qubit must control anything is on the \emph{phase angles} of the controlled-phase operator $C^{00\cdots 0}_{\Phi}$ inside the alternating sequence (Def.~\ref{def:Alternating phase modulation sequence}). 
  Controlling these angles per segment contributes an additive $Q_{\max}G$ overhead, where $Q_{\max}$ is the largest polynomial degree used by any piece.
\end{enumerate}
\end{remark}

\end{document}